\documentclass[draftcls,onecolumn, 12pt]{IEEEtran}
\ifx\notloadhyperref\undefined
\usepackage{xr-hyper}
\else
\usepackage{xr}
\fi

\usepackage{silence}
\usepackage{listings}
\usepackage{babel}
\usepackage{lipsum}
\usepackage{forest}

\usepackage[T1]{fontenc}
\usepackage[utf8]{inputenc}
\usepackage{mathtools}

\usepackage{amssymb,mathrsfs}% Typical maths resource packages
\usepackage{amsthm}
\usepackage{bm}
\usepackage{scalerel}
\usepackage{nicefrac}
\usepackage{microtype} 
\usepackage[shortlabels]{enumitem}
\usepackage{graphicx}
\usepackage{epstopdf}
\DeclareGraphicsExtensions{.eps,.png,.jpg,.pdf}

\usepackage{url}
\usepackage{colortbl}
\usepackage{booktabs}
\usepackage{multirow}
\usepackage{colortbl,xcolor}
\usepackage[normalem]{ulem}
\usepackage{xparse,xstring}
\usepackage{calc}
\usepackage{etoolbox}

\makeatletter
\@ifpackageloaded{natbib}{
	\relax
}{
	\usepackage{cite}
}
\makeatother

\usepackage{array}
\newcolumntype{L}[1]{>{\raggedright\let\newline\\\arraybackslash\hspace{0pt}}m{#1}}
\newcolumntype{C}[1]{>{\centering\let\newline\\\arraybackslash\hspace{0pt}}m{#1}}
\newcolumntype{R}[1]{>{\raggedleft\let\newline\\\arraybackslash\hspace{0pt}}m{#1}}

\makeatletter
\let\MYcaption\@makecaption
\makeatother
\usepackage[font=footnotesize]{subcaption}
\makeatletter
\let\@makecaption\MYcaption
\makeatother

\usepackage{glossaries}
\makeatletter
\sfcode`\.1006

\let\oldgls\gls
\let\oldglspl\glspl

\newcommand\fussy@ifnextchar[3]{%
	\let\reserved@d=#1%
	\def\reserved@a{#2}%
	\def\reserved@b{#3}%
	\futurelet\@let@token\fussy@ifnch}
\def\fussy@ifnch{%
	\ifx\@let@token\reserved@d
		\let\reserved@c\reserved@a
	\else
		\let\reserved@c\reserved@b
	\fi
	\reserved@c}

\renewcommand{\gls}[1]{%
\oldgls{#1}\fussy@ifnextchar.{\@checkperiod}{\@}}
\renewcommand{\glspl}[1]{%
\oldglspl{#1}\fussy@ifnextchar.{\@checkperiod}{\@}}

\newcommand{\@checkperiod}[1]{%
	\ifnum\sfcode`\.=\spacefactor\else#1\fi
}

\robustify{\gls}
\robustify{\glspl}
\makeatother

\newacronym{wrt}{w.r.t.}{with respect to}
\newacronym{RHS}{R.H.S.}{right-hand side}
\newacronym{LHS}{L.H.S.}{left-hand side}
\newacronym{iid}{i.i.d.}{independent and identically distributed}
\newacronym{SOTA}{SOTA}{state-of-the-art}
\usepackage{float}

\ifx\notloadhyperref\undefined
	\ifx\loadbibentry\undefined
		\usepackage[hidelinks,hypertexnames=false]{hyperref}
	\else
		\usepackage{bibentry}
		\makeatletter\let\saved@bibitem\@bibitem\makeatother
		\usepackage[hidelinks,hypertexnames=false]{hyperref}
		\makeatletter\let\@bibitem\saved@bibitem\makeatother
	\fi
\else
	\ifx\loadbibentry\undefined
		\relax
	\else
		\usepackage{bibentry}
	\fi
\fi

\usepackage[capitalize]{cleveref}
\crefname{equation}{}{}
\Crefname{equation}{}{}
\crefname{claim}{claim}{claims}
\crefname{step}{step}{steps}
\crefname{line}{line}{lines}
\crefname{condition}{condition}{conditions}
\crefname{dmath}{}{}
\crefname{dseries}{}{}
\crefname{dgroup}{}{}

\crefname{Problem}{Problem}{Problems}
\crefformat{Problem}{Problem~#2#1#3}
\crefrangeformat{Problem}{Problems~#3#1#4 to~#5#2#6}

\crefname{Theorem}{Theorem}{Theorems}
\crefname{Corollary}{Corollary}{Corollaries}
\crefname{Proposition}{Proposition}{Propositions}
\crefname{Lemma}{Lemma}{Lemmas}
\crefname{Definition}{Definition}{Definitions}
\crefname{Example}{Example}{Examples}
\crefname{Assumption}{Assumption}{Assumptions}
\crefname{Remark}{Remark}{Remarks}
\crefname{Rem}{Remark}{Remarks}
\crefname{remarks}{Remarks}{Remarks}
\crefname{Appendix}{Appendix}{Appendices}
\crefname{Supplement}{Supplement}{Supplements}
\crefname{Exercise}{Exercise}{Exercises}
\crefname{Theorem_A}{Theorem}{Theorems}
\crefname{Corollary_A}{Corollary}{Corollaries}
\crefname{Proposition_A}{Proposition}{Propositions}
\crefname{Lemma_A}{Lemma}{Lemmas}
\crefname{Definition_A}{Definition}{Definitions}

\usepackage{crossreftools}
\ifx\notloadhyperref\undefined
\else
	\relax
\fi

\usepackage{algorithm}
\usepackage{algpseudocode}

\ifx\loadbreqn\undefined
	\relax
\else
	\usepackage{breqn}
\fi

\makeatletter
\def\cleartheorem#1{%
    \expandafter\let\csname#1\endcsname\relax
    \expandafter\let\csname c@#1\endcsname\relax
}
\def\clearthms#1{ \@for\tname:=#1\do{\cleartheorem\tname} }
\makeatother

\ifx\renewtheorem\undefined
	\ifx\useTheoremCounter\undefined
		\newtheorem{Theorem}{Theorem}
		\newtheorem{Corollary}{Corollary}
		\newtheorem{Proposition}{Proposition}
		
	\else
		\newtheorem{Theorem}{Theorem}

	\fi

	\newtheorem{Definition}{Definition}
	\newtheorem{Example}{Example}
	
	\newtheorem{Assumption}{Assumption}

\fi

\theoremstyle{remark}

\theoremstyle{plain}

\newcommand{\qednew}{\nobreak \ifvmode \relax \else
		\ifdim\lastskip<1.5em \hskip-\lastskip
			\hskip1.5em plus0em minus0.5em \fi \nobreak
		\vrule height0.75em width0.5em depth0.25em\fi}



\NewDocumentCommand{\movedownsub}{e{^_}}{%
	\IfNoValueTF{#1}{%
		\IfNoValueF{#2}{^{}}% neither ^ nor _, do nothing; if no ^ but _, add ^{}
	}{%
		^{#1}% add superscript if present
	}%
	\IfNoValueF{#2}{_{#2}}% add subscript if present
}

\let\latexchi\chi
\RenewDocumentCommand{\chi}{}{\latexchi\movedownsub}

\newcommand{\Real}{\mathbb{R}}

\newcommand{\calC}{\mathcal{C}}

\newcommand{\calS}{\mathcal{S}}

\newcommand{\calW}{\mathcal{W}}
\newcommand{\calX}{\mathcal{X}}
\newcommand{\calY}{\mathcal{Y}}

\newcommand{\bT}{\mathbf{T}}

\newcommand{\bbN}{\mathbb{N}}

\newcommand{\bbZ}{\mathbb{Z}}

\newcommand{\scS}{\mathscr{S}}

\DeclareSymbolFont{bsfletters}{OT1}{cmss}{bx}{n}
\DeclareSymbolFont{ssfletters}{OT1}{cmss}{m}{n}
\DeclareMathSymbol{\bsfGamma}{0}{bsfletters}{'000}
\DeclareMathSymbol{\ssfGamma}{0}{ssfletters}{'000}
\DeclareMathSymbol{\bsfDelta}{0}{bsfletters}{'001}
\DeclareMathSymbol{\ssfDelta}{0}{ssfletters}{'001}
\DeclareMathSymbol{\bsfTheta}{0}{bsfletters}{'002}
\DeclareMathSymbol{\ssfTheta}{0}{ssfletters}{'002}
\DeclareMathSymbol{\bsfLambda}{0}{bsfletters}{'003}
\DeclareMathSymbol{\ssfLambda}{0}{ssfletters}{'003}
\DeclareMathSymbol{\bsfXi}{0}{bsfletters}{'004}
\DeclareMathSymbol{\ssfXi}{0}{ssfletters}{'004}
\DeclareMathSymbol{\bsfPi}{0}{bsfletters}{'005}
\DeclareMathSymbol{\ssfPi}{0}{ssfletters}{'005}
\DeclareMathSymbol{\bsfSigma}{0}{bsfletters}{'006}
\DeclareMathSymbol{\ssfSigma}{0}{ssfletters}{'006}
\DeclareMathSymbol{\bsfUpsilon}{0}{bsfletters}{'007}
\DeclareMathSymbol{\ssfUpsilon}{0}{ssfletters}{'007}
\DeclareMathSymbol{\bsfPhi}{0}{bsfletters}{'010}
\DeclareMathSymbol{\ssfPhi}{0}{ssfletters}{'010}
\DeclareMathSymbol{\bsfPsi}{0}{bsfletters}{'011}
\DeclareMathSymbol{\ssfPsi}{0}{ssfletters}{'011}
\DeclareMathSymbol{\bsfOmega}{0}{bsfletters}{'012}
\DeclareMathSymbol{\ssfOmega}{0}{ssfletters}{'012}

\makeatletter
\newcommand*\rel@kern[1]{\kern#1\dimexpr\macc@kerna}
\newcommand*\widebar[1]{%
  \begingroup
  \def\mathaccent##1##2{%
    \rel@kern{0.8}%
    \overline{\rel@kern{-0.8}\macc@nucleus\rel@kern{0.2}}%
    \rel@kern{-0.2}%
  }%
  \macc@depth\@ne
  \let\math@bgroup\@empty \let\math@egroup\macc@set@skewchar
  \mathsurround\z@ \frozen@everymath{\mathgroup\macc@group\relax}%
  \macc@set@skewchar\relax
  \let\mathaccentV\macc@nested@a
  \macc@nested@a\relax111{#1}%
  \endgroup
}
\makeatother

\DeclareMathOperator{\var}{var}

\DeclareMathOperator{\cov}{cov}

\newcommand{\ifbcdot}[1]{\ifblank{#1}{\cdot}{#1}}

\DeclarePairedDelimiterX\abs[1]{\lvert}{\rvert}{\ifbcdot{#1}}
\DeclarePairedDelimiterX\parens[1]{(}{)}{\ifbcdot{#1}}
\DeclarePairedDelimiterX\brk[1]{[}{]}{\ifbcdot{#1}}
\DeclarePairedDelimiterX\braces[1]{\{}{\}}{\ifbcdot{#1}}
\DeclarePairedDelimiterX\angles[1]{\langle}{\rangle}{\ifblank{#1}{\cdot,\cdot}{#1}}
\DeclarePairedDelimiterX\ip[2]{\langle}{\rangle}{\ifbcdot{#1},\ifbcdot{#2}}
\DeclarePairedDelimiterX\norm[1]{\lVert}{\rVert}{\ifbcdot{#1}}
\DeclarePairedDelimiterX\ceil[1]{\lceil}{\rceil}{\ifbcdot{#1}}
\DeclarePairedDelimiterX\floor[1]{\lfloor}{\rfloor}{\ifbcdot{#1}}

\DeclareFontFamily{U}{matha}{\hyphenchar\font45}
\DeclareFontShape{U}{matha}{m}{n}{
      <5> <6> <7> <8> <9> <10> gen * matha
      <10.95> matha10 <12> <14.4> <17.28> <20.74> <24.88> matha12
      }{}
\DeclareSymbolFont{matha}{U}{matha}{m}{n}
\DeclareFontSubstitution{U}{matha}{m}{n}

\DeclareFontFamily{U}{mathx}{\hyphenchar\font45}
\DeclareFontShape{U}{mathx}{m}{n}{
      <5> <6> <7> <8> <9> <10>
      <10.95> <12> <14.4> <17.28> <20.74> <24.88>
      mathx10
      }{}
\DeclareSymbolFont{mathx}{U}{mathx}{m}{n}
\DeclareFontSubstitution{U}{mathx}{m}{n}

\DeclareMathDelimiter{\vvvert}{0}{matha}{"7E}{mathx}{"17}
\DeclarePairedDelimiterX\vertiii[1]{\vvvert}{\vvvert}{\ifbcdot{#1}}

\DeclarePairedDelimiterXPP\trace[1]{\operatorname{Tr}}{(}{)}{}{\ifbcdot{#1}} % column vector
\DeclarePairedDelimiterXPP\col[1]{\operatorname{col}}{\{}{\}}{}{\ifbcdot{#1}} % column vector
\DeclarePairedDelimiterXPP\row[1]{\operatorname{row}}{\{}{\}}{}{\ifbcdot{#1}} % row vector
\DeclarePairedDelimiterXPP\erf[1]{\operatorname{erf}}{(}{)}{}{\ifbcdot{#1}}
\DeclarePairedDelimiterXPP\erfc[1]{\operatorname{erfc}}{(}{)}{}{\ifbcdot{#1}}
\DeclarePairedDelimiterXPP\KLD[2]{D}{(}{)}{}{\ifbcdot{#1}\, \delimsize\|\, \ifbcdot{#2}} % KL divergence
\DeclarePairedDelimiterXPP\op[2]{\operatorname{#1}}{(}{)}{}{#2} % general operator

\newcommand{\ud}{\,\mathrm{d}} % for integrals like \int f(x) \ud x
\DeclarePairedDelimiterXPP\indicate[1]{{\bf 1}}{\{}{\}}{}{\ifbcdot{#1}}

\NewDocumentCommand\ofrac{s m}{%
	\IfBooleanTF#1%
	{\dfrac{1}{#2}}%
	{\frac{1}{#2}}%
}
\NewDocumentCommand\ddfrac{s m m}{%
	\IfBooleanTF#1%
	{\dfrac{\mathrm{d} {#2}}{\mathrm{d} {#3}}}%
	{\frac{\mathrm{d} {#2}}{\mathrm{d} {#3}}}%
}
\NewDocumentCommand\ppfrac{s m m}{%
	\IfBooleanTF#1%
	{\dfrac{\partial {#2}}{\partial {#3}}}%
	{\frac{\partial {#2}}{\partial {#3}}}%
}

\providecommand\given{}
\DeclarePairedDelimiterX\Set[2]\{\}{%
\renewcommand\given{\SetSymbol[\delimsize]{#1}}
#2
}
\DeclarePairedDelimiterX\Setc[1]\{\}{%
\renewcommand\given{\SetSymbol{:}}
#1
}

\NewDocumentCommand\set{s o m}{%
	\IfBooleanTF#1%
	{\IfValueTF{#2}{\Set*{#2}{#3}}{\Setc*{#3}}}%
	{\IfValueTF{#2}{\Set{#2}{#3}}{\Setc{#3}}}%
}

\NewDocumentCommand{\evalat}{ s O{\big} m e{_^} }{%
\IfBooleanTF{#1}%
{\left. #3 \right|}{#3#2|}%
\IfValueT{#4}{_{#4}}%
\IfValueT{#5}{^{#5}}%
}

\providecommand\given{}
\DeclarePairedDelimiterXPP\cprob[1]{}(){}{
\renewcommand\given{\nonscript\,\delimsize\vert\allowbreak\nonscript\,\mathopen{}}%
#1%
}
\DeclarePairedDelimiterXPP\cexp[1]{}[]{}{
\renewcommand\given{\nonscript\,\delimsize\vert\allowbreak\nonscript\,\mathopen{}}%
#1%
}

\DeclareDocumentCommand \P { s e{_^} d() g } {%
	\mathbb{P}%
	\IfBooleanTF{#1}%
		{
			\IfValueT{#2}{_{#2}}%
			\IfValueT{#3}{^{#3}}%
			\IfValueTF{#5}{\cprob{#4 \given #5}}{\IfValueT{#4}{\cprob{#4}}}%
		}%
		{
			\IfValueT{#2}{_{#2}}%
			\IfValueT{#3}{^{#3}}%
			\IfValueTF{#5}{\cprob*{#4 \given #5}}{\IfValueT{#4}{\cprob*{#4}}}%
		}%
}

\DeclareDocumentCommand \E { s e{_^} o g } {%
	\mathbb{E}%
	\IfBooleanTF{#1}%
		{
			\IfValueT{#2}{_{#2}}%
			\IfValueT{#3}{^{#3}}%
			\IfValueTF{#5}{\cexp{#4 \given #5}}{\IfValueT{#4}{\cexp{#4}}}%
		}%
		{
			\IfValueT{#2}{_{#2}}%
			\IfValueT{#3}{^{#3}}%	
			\IfValueTF{#5}{\cexp*{#4 \given #5}}{\IfValueT{#4}{\cexp*{#4}}}%		
		}%
}

\DeclareDocumentCommand \Var { s e{_^} d() g } {%
	\var%
	\IfBooleanTF{#1}%
		{
			\IfValueT{#2}{_{#2}}%
			\IfValueT{#3}{^{#3}}%
			\IfValueTF{#5}{\cprob{#4 \given #5}}{\IfValueT{#4}{\cprob{#4}}}%
		}%
		{
			\IfValueT{#2}{_{#2}}%
			\IfValueT{#3}{^{#3}}%	
			\IfValueTF{#5}{\cprob*{#4 \given #5}}{\IfValueT{#4}{\cprob*{#4}}}%		
		}%
}

\DeclareDocumentCommand \Cov { s e{_^} d() g } {%
	\cov%
	\IfBooleanTF{#1}%
		{
			\IfValueT{#2}{_{#2}}%
			\IfValueT{#3}{^{#3}}%
			\IfValueTF{#5}{\cprob{#4 \given #5}}{\IfValueT{#4}{\cprob{#4}}}%
		}%
		{
			\IfValueT{#2}{_{#2}}%
			\IfValueT{#3}{^{#3}}%	
			\IfValueTF{#5}{\cprob*{#4 \given #5}}{\IfValueT{#4}{\cprob*{#4}}}%		
		}%
}

\ExplSyntaxOn
\NewDocumentCommand \dist {m o o} {%
\mathrm{#1}\left(%
	\IfValueT{#3}{%
		\tl_if_blank:nTF{ #3 }{\cdot\, \middle|\, }{#3\, \middle|\, }%
	}
	\IfValueT{#2}{#2}%
\right)%
}
\ExplSyntaxOff

\NewDocumentCommand {\cbrace} {t+ D[]{black} D(){\widthof{#5}} m m } {%
	\begingroup%
		\color{#2}
		\IfBooleanTF{#1}{%
			\overbrace{#4}^%
		}{
			\underbrace{#4}_%
		}%
		{\parbox[c]{#3}{\centering\footnotesize{#5}}}%
	\endgroup% 
}

\let\oldforall\forall
\renewcommand{\forall}{\oldforall \, }

\let\oldexist\exists
\renewcommand{\exists}{\oldexist \, }

\makeatletter

\newcommand{\rankcolor}[2]{%
	\expandafter\renewcommand\csname #1\endcsname[1]{%
		\ifblank{##1}{%
			{\color{#2} \textbf{#2}}%
		}{%
			\ifmmode
				\textcolor{#2}{\bm{##1}}%
			\else%
				{\color{#2} \textbf{##1}}%
			\fi	
		}%
	}
}

\rankcolor{first}{red}
\rankcolor{second}{blue}
\rankcolor{third}{cyan}
\makeatother

\graphicspath{{./Figures/}{./figures/}}
\DeclareDocumentCommand{\includeCroppedPdf}{ o O{./Figures/} m }{
	\IfFileExists{#2#3-crop.pdf}{}{%
		\immediate\write18{pdfcrop #2#3.pdf #2#3-crop.pdf}}%
	\includegraphics[#1]{#2#3-crop.pdf}
}

\makeatletter
\newcommand*{\addFileDependency}[1]{% argument=file name and extension
  \typeout{(#1)}
  \@addtofilelist{#1}
  \IfFileExists{#1}{}{\typeout{No file #1.}}
}
\makeatother

\definecolor{gray90}{gray}{0.9}
\def\colorlist{red,blue,brown,cyan,darkgray,gray,lightgray,green,lime,magenta,olive,orange,pink,purple,teal,violet,white,yellow}

\makeatletter
\def\startcomment{[}
\ifx\nohighlights\undefined
	\newcommand{\createcolor}[1]{%
			\expandafter\newcommand\csname #1\endcsname[1]{{\color{#1} ##1}}%
	}
	\newcommand{\msout}[1]{\text{\color{green} \sout{\ensuremath{#1}}}}
	\newcommand{\del}[1]{{\color{green}\ifmmode \msout{#1}\else\sout{#1}\fi}}
\else
	\newcommand{\createcolor}[1]{%
			\expandafter\newcommand\csname #1\endcsname[1]{%
				\noexpandarg%
				\StrChar{##1}{1}[\firstletter]%
				\if\firstletter\startcomment%
					\relax
				\else%
					##1
				\fi
			}%
	}
	\newcommand{\msout}[1]{}
	\newcommand{\del}[1]{}
\fi

\def\@tempa#1,{%
    \ifx\relax#1\relax\else
        \createcolor{#1}%
        \expandafter\@tempa
    \fi
}
\expandafter\@tempa\colorlist,\relax,
\makeatother

\newcommand{\hhide}[1]{}
\ifx\diagnoselabel\undefined
	\relax
\else
	\makeatletter
	\def\@testdef #1#2#3{%
		\def\reserved@a{#3}\expandafter \ifx \csname #1@#2\endcsname
			\reserved@a  \else
			\typeout{^^Jlabel #2 changed:^^J%
				\meaning\reserved@a^^J%
				\expandafter\meaning\csname #1@#2\endcsname^^J}%
			\@tempswatrue \fi}
	\makeatother
\fi

\newcommand\restr[2]{{% we make the whole thing an ordinary symbol
		\left.\kern-\nulldelimiterspace % automatically resize the bar with \right
		#1 % the function
		\littletaller % pretend it's a little taller at normal size
		\right|_{#2} % this is the delimiter
}}

\newcommand{\littletaller}{\mathchoice{\vphantom{\big|}}{}{}{}}

\newacronym{GFT}{GFT}{graph Fourier transform}
\newacronym{GSP}{GSP}{graph signal processing}
\newacronym{GNN}{GNN}{graph neural network}
\newacronym{GSO}{GSO}{graph shift operator}
\newacronym{MSE}{MSE}{mean-squared error}

\begin{document}
	
	\title{Generalized Graphon Process: Convergence of Graph Frequencies in Stretched Cut Distance}
	\author{Xingchao Jian, Feng Ji, Wee~Peng~Tay,~\IEEEmembership{Senior~Member,~IEEE}%
		% \thanks{
			% 	This work was supported by the xxxxx. Preliminary versions of parts of this paper were presented at xxxx.  The author is with the School of Electrical and Electronic Engineering, Nanyang Technological University, Singapore. E-mail: \texttt{wptay@ntu.edu.sg}
			% }%
	}
	
	%\markboth{IEEE TRANSACTIONS ON INFORMATION THEORY,~Vol.~, No.~, ~2012}%
	%{Tay: xxxx}
	
	% The paper headers
	%\markboth
	%   {To be submitted... }
	%   {Tay \MakeLowercase{\textit{et al.}}: }
	
	% make the title area
	% Don't write page number 0 to the cover page.
	\maketitle \thispagestyle{empty}
	
	% Put abstract and the paper's body in a new page, page 1.
	%\newpage
	%\setcounter{page}{1}
	
	%---------------------------------------------------------------------------%
	%                           abstract and key words                          %
	%---------------------------------------------------------------------------%
	\begin{abstract}
		Graphons have traditionally served as limit objects for dense graph sequences, with the cut distance serving as the metric for convergence. 
		However, sparse graph sequences converge to the trivial graphon under the conventional definition of cut distance, which make this framework inadequate for many practical applications. In this paper, we utilize the concepts of generalized graphons and stretched cut distance to describe the convergence of sparse graph sequences. Specifically, we consider a random graph process generated from a generalized graphon. This random graph process converges to the generalized graphon in stretched cut distance. We use this random graph process to model the growing sparse graph, and prove the convergence of the adjacency matrices' eigenvalues. We supplement our findings with experimental validation. Our results indicate the possibility of transfer learning between sparse graphs.
	\end{abstract}
	
	\begin{IEEEkeywords}
		Generalized graphons, sparse graph sequence, convergent graph frequencies.
	\end{IEEEkeywords}
	
	\section{Introduction}\label{sect:intro}
	
	Modern data analysis usually involves complex structures like graphs. In order to model and process signals on graphs, the \gls{GSP} has established a set of tools for a variety of tasks, including sampling, reconstruction and filtering \cite{OrtFroKov:J18,TanEldOrt:J20,JiaJiTay:J23}. Besides, by introducing non-linearity, \gls{GNN} provides a deep learning architecture and has been largely studied. These methods usually have good performances by exploiting the graph information when the underlying graph structure is known. In addition, they usually have good computational properties such as distributed implementation \cite{SegMarRib:J17} and robustness to perturbation \cite{CecBar:J20,SongKangWang:C22}. 
	
	In practice, designing signal processing techniques separately on different graphs can be computationally expensive. For example, in order to learn a graph filter, the eigendecomposition of \gls{GSO} can be computationally prohibited when the graph is large. Therefore, it is natural to consider learning graph filter or \gls{GNN} on a graph with a small or moderate size and then transfer it to other graphs which can be large. The success of such strategies relies on the inherent similarity between the graph used for training and testing. For example, the paper \cite{LevHuaBuc:J21} studied the transferability of \gls{GNN} by modeling the graphs as down-sampled version of a topological space, and the graph signals as samples of function on the space. The paper \cite{RonIsuKut:C19} studied the transferability of graph filters in Cayley smoothness space.
	
	In recent years, an arising way to explain the transferability of graph filter and \gls{GNN} is through the \emph{graphon} method \cite{Lov:12,RuiChaRib:J21}. The graphon method can be understood in two folds: 
	i) as the number of nodes tends to infinity, it is assumed that the graph sequence will converge to a graphon in cut distance, i.e., graphon is a \emph{limit object} of the graph sequence. This was proved to imply the graph frequency convergence \cite{RuiChaRib:J21}. 
	ii) the graphs of interest are generated from the same probabilistic model (graphon). This model was utilized to bound the difference between the outputs of a \gls{GNN} with fixed parameters on two different graphs sampled from the same graphon \cite{RuiChaRib:C20,RuiChaRib:C21}.
	
	Graphon is suitable for modeling the limit of dense graph sequences under the cut distance. However, if the graph sequence is sparse, then it is no longer appropriate. By saying a graph sequence $(G_n)_{n\geq1}$ is \emph{sparse}, we mean that $\lim\limits_{n\to\infty}\dfrac{|E(G_n)|}{|V(G_n)|^2}= 0$, where $V(G_n)$ and $E(G_n)$ are the vertex and edge sets of $G_n$. In this case, the cut norm of $(G_n)_{n\geq1}$ converges to $0$ since it equals $\dfrac{2|E(G_n)|}{|V(G_n)|^2}$, i.e., all sparse sequences $(G_n)_{n\geq1}$ converge to the zero graphon if we use the standard definitions of graphon and cut distance. Therefore, in order to discuss transferability of filters or \gls{GNN}s on sparse graphs, we need alternative concepts of graphon and cut distance for sparse graph sequences. These concepts are the main focus of this paper. Our main contributions are:
	%The graph sparsity can be due to different mechanisms. For example, the paper \cite{RodGamBar:J22} studied the case when the graphs of interest are of bounded degree. The paper \cite{BorChaCoh:J18} studied the case when the graphs of interest have uniformly regular tails.
	\begin{enumerate}
		\item We introduce the notions of generalized graphon and stretched cut distance for sparse graph convergence in place of the standard graphon and cut distance, which are more suitable for dense graph convergence. In particular, we introduce the notion of a generalized graphon process associated with the generalized graphon. This process converges to the generalized graphon in stretched cut distance. We model a sparse graph sequence as a subsequence of this random graph process.
		\item We prove that under the generalized graphon process model, the graph frequencies of the process have a linear relationship with the square root of the number of edges as the graph size grows asymptotically. 
		\item We compare the fitness of our theories and the standard graphon's theories on real dataset to show better fitness of the generalized graphon process model and the correctness of our theoretical result. 
	\end{enumerate}
	
	The rest of this paper is organized as follows. In \cref{sect:genmodel} we introduce a graph generating process based on a generalized notion of graphon. In \cref{sect:conv_result} we prove the convergence of graph frequencies of this process. In \cref{sect:exp} we corroborate our result by numerical experiments. We conclude the paper in \cref{sect:conc}. 
	
	\emph{Notations.} For any set $A$, we use $I_A$ to denote the indicator function on it.  We write $\Real_+$ as the set of non-negative real numbers. We denote cut norm by $\norm{\cdot}_{\square}$. For two functions $f_1$ and $f_2$, we write their composition as $f_1\circ f_2$. For a function $f:\calX\to\calY$, we define
	\begin{align*}
		f\times f:\calX\times\calX&\to\calY\times\calY \\
		(x_1,x_2)&\mapsto (f(x_1),f(x_2)).
	\end{align*}
	For two sets $\calX_1$ and $\calX_2$, we define the projection $\pi_i$ ($i=1,2$) as 
	\begin{align*}
		\pi_i:\calX_1 \times \calX_2 &\to \calX_i \\
		(x_1,x_2) &\mapsto x_i.
	\end{align*}

	\section{Generalized Graphon Process and Stretched Cut Distance}\label{sect:genmodel}
	
	In this section, we introduce the concepts of generalized graphon and generalized graphon process as a generating model for random sparse graph sequences. We then introduce the stretched cut distance to characterize the convergence of sparse graph sequences. 
	
	\subsection{Generalized Graphon and Graphon Process}
	In this subsection, we introduce a generalized definition of graphon, and an associated graph generating model studied in \cite{BorChaCoh:J18}. We assume that the growing graphs are generated by the following two components: an underlying feature space $\scS=(S, \calS, \mu)$ which is a $\sigma$-finite measure space and a symmetric function $W:S\times S\mapsto [0,1]$ such that $W\in L^2(S\times S)$. As we will see in the ensuing content, $\scS$ contains the feature that will be utilized by $W$ to determine the edges between the vertices of an infinite graph. We refer to the tuple $\calW = (W,\scS)$ as \emph{generalized graphon}. In \cite{BorChaCoh:J19} the most general exchangeable random graph model contains two more components representing isolated and star structures. Here for ease of analysis we omit them and consider the version in \cite{BorChaCoh:J18}, i.e., $\calW = (W,\scS)$. Note that when $\scS=[0,1]$, $\calW$ is the graphon commonly used in the exsiting graphon signal processing literature \cite{RuiChaRib:J21,RuiCha:C20}. We refer to these specific graphons as \emph{standard graphons}. In the rest of this paper, we always make the following assumption unless otherwise stated.
	
	\begin{Assumption}\label{asp:meas_space}
		The measure spaces $\scS$ under consideration are $\sigma$-finite, Borel, and atom-less.
	\end{Assumption}

	In order to model the growing process of graphs, we need to introduce the time dimension encoded by $\Real_+$. To be specific, we assign a Poisson point process $\Gamma$ on $\Real_+\times S$. We denote each point of this point process as $v=(t,x)$, where $t\in\Real_+$ denotes time and $x\in S$ denotes feature. The process $\Gamma$ induces an infinite graph $\tilde{G}$ with vertex set $V$ the set of all points generated by $\Gamma$. The edges of $\tilde{G}$ are randomly generated such that two different vertices $u = (t,x)$ and $u' = (t',x')$ are connected with probability $W(x,x')$. We do not assign any self-loop on $\tilde{G}$. From the nature of the generating process the graph $\tilde{G}$ can be regarded as containing all vertices that will arise. Then the growing graph at time instance $t$ will be $\tilde{G}_t$ which is the subgraph of $\tilde{G}$ with the vertex set $\tilde{V}_t = \set{(t',x)\in V \given t'\leq t}$. We denote $G_t$ as the graph obtained by removing all isolated vertices from $\tilde{G}_t$. In this paper, we will mainly focus on the sequence $(G_t)_{t\geq0}$, and refer to it as \emph{generalized graphon process}. We write $G_t=(V_t, E_t)$ where $|V_t| := N_t$. According to \cite{BorChaCoh:J18}, $(G_t)_{t\geq0}$ converge to $\calW$ in stretched cut distance, hence is suitable for modeling sparse growing networks (see \cref{subsect:stretch} for details).
	
	Under \cref{asp:meas_space}, we know that $\scS$ is isomorphic to $[0, \mu(S))$ by \cite[Lemma 33]{BorChaCoh:J18}. Therefore, if $\mu(S)<\infty$, then $\scS$ can be identified with a bounded interval. In this paper, we allow $\mu(S)=\infty$, in which case $\scS$ can be identified with $\Real_+$.
	%We say a graph sequence $(G_n)$ is \emph{sparse} if $\lim\limits_{n\to\infty}\dfrac{|E(G_n)|}{|V(G_n)|^2}= 0$.
	
	In graphon theories, the graphs can be associated with a canonical graphon \cite[Section 7.1]{Lov:12}:
	
	\begin{Definition}
		Given a finite simple graph $G$ with vertex set $\set{v_i\given i=1,\dots,n}$ and edge set $E$. The canonical graphon associated with $G$ is defined as a step function
		\begin{align}
			W^{G} = \sum_{(v_i,v_j)\in E} I_{A_{ij}},
		\end{align}
		where $A_{ij}$ is the square $\left[\dfrac{i-1}{n},\dfrac{i}{n}\right) \times  \left[\dfrac{j-1}{n},\dfrac{j}{n}\right)$. We write $\calW^G:=(W^{G},\Real_+)$.
	\end{Definition}
	Note that $W^{G}$ can vary with the labeling of the vertex set, but this variation will not affect the evaluation of cut distance. Here we remind that the cut norm $\norm{\cdot}_\square$ and cut distance $\delta_{\square}$ are defined as \cite[Definition 5]{BorChaCoh:J18}
	\begin{align*}
		\norm{W}_{\square,S,\mu} &= \sup_{U,V\in\calS} \abs*{\int_{U\times V} W(x,y)\ud\mu(x)\ud\mu(y)},\\
		\delta_{\square}(\calW_1,\calW_2) &= \inf_{\tilde{\mu}\in\calC(\mu_1,\mu_2)} \norm*{W_1\circ(\pi_1\times\pi_1) - W_2\circ(\pi_2\times\pi_2)}_{\square,S_1\times S_2,\tilde{\mu}},
	\end{align*}
	where $\calC(\mu_1,\mu_2)$ is the set of all couplings of $\mu_1$ and $\mu_2$.
	For a measure $\mu$ over $(S_1\times S_2, \calS_1\times\calS_2)$, if $\mu$ has $\mu_1$ and $\mu_2$ as its marginal, then $\mu$ is called a \emph{coupling} of $\mu_1$ and $\mu_2$. Mathematically, this means $\mu(A\times S_2) = \mu_1(A)$ for all $A\in\calS_1$ and $\mu(S_1\times B) = \mu_2(B)$ for all $B\in\calS_2$. We omit the notions of $S$ and $\mu$ in the subscript of cut norm if they are clear in the context.
	%Alternatively, under \cref{asp:meas_space}, since we can identify $\scS_1$ and $\scS_2$ as $\Real_+$, the cut distance can be equivalently defined as \cite[Proposition 48]{BorChaCoh:J18}
	%\begin{align}\label{eq:perm_cut_dist}
	%	\delta_{\square}(W_1,W_2) &= \inf_{\tilde{\sigma}} \norm{W_1 - W_2 \circ (\tilde{\sigma}\times\tilde{\sigma})}_\square,
	%\end{align}
	%where $\tilde{\sigma}$ is taken over all interval permutations on $\Real_+$.
	
	\subsection{Stretched Cut Distance}\label{subsect:stretch}
	
	In this subsection, we introduce a rescaled version of cut distance to describe the convergence of sparse graph sequence. It has been explained in \cref{sect:intro} that the notion of standard cut distance is not able to describe the convergence in the sparse setting. In order to find reasonable limit of sparse graph sequence, the \emph{stretched cut distance} is defined as the cut distance between rescaled versions of graphons: 
	
	\begin{Definition}\cite[Definition 11]{BorChaCoh:J18}\label{def:stretch}
		For a graphon $\calW = (W,\scS)$, define $\calW^s = (W,\hat{\scS})$ where $\hat{\scS} = (S,\calS,\hat{\mu}) = (S,\calS,\norm{W}_1^{-\ofrac{2}}\mu)$.  We refer to $\calW^s$ as \emph{stretched graphon}. The stretched cut distance between two graphons $\calW_1$ and $\calW_2$ is defined as $\delta_{\square}^s(\calW_1,\calW_2):= \delta_{\square}(\calW_1^s,\calW_2^s)$. A sequence of graphs $(G_n)$ is called convergent to a graphon $\calW$ if their canonical graphons $(W^{G_n})$ converges to $\calW$ in stretched cut distance, i.e., $\lim\limits_{n\to\infty} \delta_{\square}^s(\calW^{G_n,s}, \calW)=0$.
	\end{Definition}
	Note that, according to the construction in \cref{def:stretch}, $\norm{W}_{\square,S,\hat{\mu}} = \norm{W}_{1,S,\hat{\mu}}\equiv1$. Therefore, it is more suitable to describe the convergence of sparse graph sequences using $\delta_{\square}^s$. It is known that the generalized graphon process $(G_t)$ generated from $\calW$ converges to $\calW$ almost surely in stretched cut distance \cite[Theorem 28]{BorChaCoh:J18}. In addition, for any graphon $\calW = (W,\Real_+)$ on $\Real_+$, if we define $\calW^{s'} := (W(\norm{W}_1^{\ofrac{2}}x_1,\norm{W}_1^{\ofrac{2}}x_2), \scS)$, then $\delta_{\square}(\calW^{s'},\calW^s)=0$ \cite{BorChaCoh:J18}. Therefore in this case we can identify $\calW^s$ with $\calW^{s'}$. Specifically, if we consider a canonical graphon $\calW^{G}$ induced by a graph $G$, then we can view $\calW^{G,s}$ as a step function on $\Real_+\times \Real_+$, in which the width and length of each square step is $\ofrac{\sqrt{2\abs{E(G)}}}$. Recall that the width and length of each square step in $\calW^G$ is $\ofrac{\abs{V(G)}}$. In the rest of the paper we will always view the stretched canonical graphon in this way. In \cref{exa:sparse_conv} we provide an example of a sparse graph sequence converging to a generalized graphon in stretched cut distance.
	
	\begin{Example}\label{exa:sparse_conv}
		Consider a sequence of graph $(G_n)$ such that $|V_n|=n$. Let $\alpha\in(0,1)$ be a constant. We choose $\floor{n^{\frac{1+\alpha}{2}}}$ vertices to form a complete subgraph, and the rest $n-\floor{n^{\frac{1+\alpha}{2}}}$ vertices are set as isolated. In this case, $\abs{E_n} = \frac{\floor{n^{\frac{1+\alpha}{2}}}(\floor{n^{\frac{1+\alpha}{2}}}-1)}{2}$. Therefore, we can label the vertices such that $\calW^{G_n,s}\equiv1$ on the region $\left[0,\frac{\floor{n^{\frac{1+\alpha}{2}}}}{\sqrt{\floor{n^{\frac{1+\alpha}{2}}}(\floor{n^{\frac{1+\alpha}{2}}}-1)}}\right)^2$, and equals $0$ elsewhere. It can be shown that $\lim\limits_{n\to\infty}\delta_{\square}(\calW^{G_n,s},I_{[0,1]^2})=0$, i.e., $(G_n)$ converges to $I_{[0,1]^2}$ in stretched cut distance. On the other hand, $(G_n)$ is a sparse graph sequence, hence will converge to a zero graphon in cut distance.
	\end{Example}
	
	\section{Convergence of Graph Frequencies}\label{sect:conv_result}
	
	In this section, we prove the convergence of graph frequencies (i.e., the eigenvalues of graph adjacency matrices) of a generalized graphon process $(G_t)$ generated from a generalized graphon $\calW$.
	
	As in the graphon literature, we consider the integral operator $\bT_\calW$ with integral kernel $W$:
	\begin{align*}
		\bT_\calW : L^2(S) &\to L^2(S) \\
		g &\mapsto \int_{S} W(x,x')g(x') \ud\mu(x').
	\end{align*}
	Since $W\in L^2(S\times S)$, the operator $\bT_\calW$ is a self-adjoint Hilbert-Schmidt operator. We denote the eigenvalues and orthonormal eigenvectors of $\bT_\calW$ as $\set{\lambda_j(\calW)\given j\in\bbZ\backslash\set{0}}$ and $\set{\varphi_j(x;\calW)\given j\in\bbZ\backslash\set{0}}$. The eigenvalues are ordered such that $\lambda_1(\calW)\geq\lambda_2(\calW)\geq\dots0$, and $\lambda_{-1}(\calW)\leq\lambda_{-2}(\calW)\leq\dots0$. Then it can be shown by \cite[Theorem 4.2.16]{Dav:07} that, $W$ can be decomposed as 
	\begin{align}\label{eq:graphex_decomp_funcs}
		W(x,x') = \sum_{j\in\bbZ\backslash\set{0}} \lambda_j(\calW)\varphi_j(x;\calW)\varphi_j(x';\calW).
	\end{align}
	%
	%The paper \cite{BorChaCoh:J18} has shown that the graph sequence $(G_t)_{t\geq0}$ converges to $W$ in stretched cut distance. 
	%In this work, we also consider the convergence of graph signal based on the generating model in \cref{sect:genmodel}. \blue{[Seems that this is not used.]} Note that the squares in the stretched graphon $W^{G_t,s}$ have width $\dfrac{1}{\sqrt{2|E(G_t)|}}$. Therefore, we define the stretched signal as 
	%\begin{align*}
	%	f_t := \sum_{(t_i,x_i)\in V_t}f(x_i)I_{B_i}, 
	%\end{align*}
	%where $B_i := [\dfrac{i-1}{\sqrt{2|E(G_t)|}}, \dfrac{i}{\sqrt{2|E(G_t)|}})$. $f_t$ also depends on the ordering of $_t$. 
	
	%In this paper, we make the following assumption: \blue{[This assumption is no longer needed. Instead, replace it with the stretched cut-norm convergence.]}
	%\begin{Assumption}\label{asp:sameorder}
	%	Assume that there exists an ordering of $_t$ such that $\norm{W^{G_t,s}-W}_{\square}\rightarrow0$ and $\norm{f_t-f}_2\rightarrow0$. 
	%\end{Assumption}
	%According to \cite[Theorem 28]{BorChaCoh:J18}, $\lim\limits_{t\to\infty} \delta_{\square}(W^{G_t,s},W)=0$. Besides, \cref{eq:perm_cut_dist} implies that $\delta_{\square}(W^{G_t,s},W)$ are achieved by interval permutations on $\Real_+$. Therefore, \cref{asp:sameorder} is essentially assuming that the cut distances $\delta_{\square}(W^{G_t,s},W)$ are achieved by permutations over the intervals $\set{B_i}$. In addition, these permutations make $(f_t)$ converge. \blue{[May also omit this.]}
	
	Provided the convergence of ${W^{G_t,s}}$, we can prove the convergence of the graph frequencies of $(G_t)$. To be specific, we will prove the convergence of the eigenvalues of ${W^{G_t,s}}$ to those of $W$ up to a scaling factor. 
	
	\begin{Theorem}\label{thm:g_freq_conv}
		Define
		\begin{align*}
			D_W(x) := \int_{S} W(x,x') \ud \mu(x') ,
		\end{align*}
		and assume that $D_W(x)\in L^p, \forall p \geq 1$. Then the graph frequencies of $(G_t)$ converges in the following way:
		\begin{align}\label{eq:g_freq_conv}
			\lim_{t\rightarrow\infty} \frac{\lambda_j(G_t)}{\sqrt{2|E(G_t)|}} = \frac{\lambda_j(W)}{\sqrt{\norm{W}_1}},
		\end{align} 
		where $\lambda_{-1}(G_t)\leq\lambda_{-2}(G_t)\leq\dots\leq0\leq\dots\leq\lambda_{2}(G_t)\leq\lambda_{1}(G_t)$ are eigenvalues of $G_t$'s adjacency matrix.
	\end{Theorem}
	
	\begin{proof}
		Given two simple graphs $F$ and $G$, we say a map $\phi:V(F)\to V(G)$ is an adjacency preserving map if $(v_i,v_j)\in E(F)$ implies $(\phi(v_i), \phi(v_j))\in E(G)$. Let $\hom(F,G)$ be the number of adjacency preserving maps between $F$ and $G$. Define
		\begin{align*}
			h(F,G) = \frac{\hom(F,G)}{(2\abs{E(G)})^{\frac{\abs{V(F)}}{2}}}.
		\end{align*}
		Generally, For a graphon $\calW$, we define 
		\begin{align*}
			h(F,\calW) = \norm{W}_1^{-\frac{\abs{V(F)}}{2}} \int_{S^{\abs{V(F)}}} \prod_{(v_i,v_j)\in E(F)} W(x_i,x_j) \ud x_1\dots\ud x_{\abs{V(F)}}.
		\end{align*}
		It can be shown that, for a canonical graphon $\calW^{G}$, we have $h(F,G) = h(F,\calW^{G})$. Besides, according to \cite[Proposition 30 (ii)]{BorChaCoh:J18}, we have $\lim\limits_{t\to\infty} h(F,G_t) = h(F,\calW)$. Therefore,
		\begin{align}\label{eq:density_conv}
			\lim_{t\to\infty} h(F,\calW^{G_t}) = h(F,\calW).
		\end{align}
		Let $F = C_k$ be a $k$-cycle, $k\geq3$. Then we have
		\begin{align}
			\begin{aligned}\label{eq:cycle-eig}
				h(C_k, \calW) &= \norm{W}_1^{-\frac{k}{2}} \int_{S^k} \prod_{(v_i,v_j)\in E(F)} W(x_i,x_j) \ud x_1\dots\ud x_{k}\\
				&= \norm{W}_1^{-\frac{k}{2}} \sum_{j\in\bbZ\backslash\set{0}} \lambda_j(\calW)^k,
			\end{aligned}
		\end{align}
		where the second equality can be obtained by replacing the integrand by \cref{eq:graphex_decomp_funcs} and using the orthogonality of $\calW$'s eigenvectors. Combining \cref{eq:density_conv} and \cref{eq:cycle-eig}, we have
		\begin{align}\label{eq:sum_on_conv}
			\lim_{t\to\infty} \sum_{j\in\bbZ\backslash\set{0}}\parens*{\frac{\lambda_j(\calW^{G_t})}{\sqrt{\norm{W^{G_t}}_1}}}^k = \sum_{j\in\bbZ\backslash\set{0}}\parens*{\frac{\lambda_j(\calW)}{\sqrt{\norm{W}_1}}}^k, \forall k\geq 3.
		\end{align}
		Note that 
		\begin{align}
			\lambda_j(\calW^{G_t}) = \frac{\lambda_j(G_t)}{\abs{V(G_t)}},
			\norm{W^{G_t}}_1 = 2\frac{\abs{E(G_t)}}{\abs{V(G_t)}^2},
		\end{align}
		hence 
		\begin{align}\label{eq:scaled_on_to_matrix}
			\frac{\lambda_j(\calW^{G_t})}{\sqrt{\norm{W^{G_t}}_1}} = \frac{\lambda_j(G_t)}{\sqrt{2\abs{E(G_t)}}}.
		\end{align}
		We next prove \cref{eq:g_freq_conv} from \cref{eq:sum_on_conv} by contradiction. In the rest of this proof, we assume there exists $k_0\in\bbZ\backslash\set{0}$ such that $\set{\dfrac{\lambda_{k_0}(\calW^{G_t})}{\sqrt{\norm{W^{G_t}}_1}}\given t\in\Real_+}$ does not converge to $\dfrac{\lambda_{k_0}(\calW)}{\sqrt{\norm{W}_1}}$ when $t\to\infty$.
		
		We first observe that $\set{\dfrac{\lambda_{j}(\calW^{G_t})}{\sqrt{\norm{W^{G_t}}_1}}\given t\in\Real_+}$ is a bounded set for all $j\in\bbZ\backslash\set{0}$. The argument goes as follows: let $k=4$. For any graphon $\calW'$ We have
		\begin{align*}
			\sum_{j=1}^m \parens*{\frac{\lambda_j(\calW')}{\sqrt{\norm{W'}_1}}}^4 \leq \sum_{j\in\bbZ\backslash\set{0}}\parens*{\frac{\lambda_j(\calW')}{\sqrt{\norm{W'}_1}}}^4 = h(C_4,\calW').
		\end{align*}
		Note that $\set{\lambda_j(\calW')\given j=1,2,\dots}$ is a non-increasing sequence. Therefore,
		\begin{align*}
			\frac{\lambda_m(\calW')}{\sqrt{\norm{W'}_1}} \leq \parens*{\frac{h(C_4,\calW')}{m}}^{\ofrac{4}}.
		\end{align*}
		Similarly, we have
		\begin{align*}
			\frac{\lambda_{-m}(\calW')}{\sqrt{\norm{W'}_1}} \geq -\parens*{\frac{h(C_4,\calW')}{m}}^{\ofrac{4}}.
		\end{align*}
		Note that $\set{h(C_4,\calW^{G_t})}$ is a convergent sequence when $t\to\infty$, hence bounded, so there exists a $B>0$ such that
		\begin{align*}
			\frac{\abs{\lambda_{j}(\calW^{G_t})}}{\sqrt{\norm{W^{G_t}}_1}} \leq \parens*{\frac{B}{\abs{j}}}^{\ofrac{4}}, \forall j \in\bbZ\backslash\set{0},
		\end{align*}
		i.e., the set $\set{\dfrac{\lambda_{j}(\calW^{G_t})}{\sqrt{\norm{W^{G_t}}_1}}\given t\in\Real_+}$ is bounded for all $j\in\bbZ\backslash\set{0}$.
		
		According to our assumption, there exists a sequence $(t_n)\to\infty$ such that $\parens*{\dfrac{\lambda_{k_0}(\calW^{G_{t_n}})}{\sqrt{\norm{W^{G_{t_n}}}_1}}}$ does not converge to $\dfrac{\lambda_{k_0}(\calW)}{\sqrt{\norm{W}_1}}$ when $n\to\infty$. For simplicity, we denote the double sequence $\parens*{\dfrac{\lambda_{j}(\calW^{G_{t_n}})}{\sqrt{\norm{W^{G_{t_n}}}_1}}}$ as $(a_{j,n})$ and write $\dfrac{\lambda_{k_0}(\calW)}{\sqrt{\norm{W}_1}}$ as $b_j$. Note that since $(a_{k_0, n})$ is bounded, we can assume without loss of generality that $\lim\limits_{n\to\infty}a_{k_0, n}$ exists and does not equal to $b_{k_0}$. We next construct a subsequence $(n_k)$ such that $\lim\limits_{k\to\infty} a_{j,n_k}$ exists for every $j\in\bbZ\backslash\set{0}$ as follows:
		\begin{enumerate}
			\item step 1: find an increasing sequence $(r_{1,i})_{i=1}^{\infty}\subset\bbN$ such that $\lim\limits_{i\to\infty} a_{1,r_{1,i}}$ exists.
			\item step 2: suppose we have constructed a sequence $(n_i)$ such that $\lim\limits_{i\to\infty} a_{j,n_i}$ exists for all $1\leq j<s$. Then we find a sequence $(r_{s,i})_{i=1}^{\infty}\subset(n_i)$ such that $\lim\limits_{i\to\infty} a_{s,r_{s,i}}$ exists.
			\item step 3: by construction of step 1 and 2, we have obtained a double sequence $(r_{j,i})$ such that $\lim\limits_{i\to\infty} a_{j,r_{j,i}}$ exists for all $j\geq 1$, and $(r_{s,i})\subset(r_{s-1,i})$. Therefore, if we consider the sequence $(r_{i,i})$, we will have $\lim\limits_{i\to\infty} a_{j,r_{i,i}}$ exists for all $j\geq 1$.
			\item step 4: by repeating the above procedure, we can find a subsequence of $(r_{i,i})$, denoted as $(n_k)$, such that $\lim\limits_{i\to\infty} a_{j,n_k}$ exists for all $j\leq -1$. This completes the construction.
		\end{enumerate}
		For simplicity, we write $(a_{j, n_k})$ as $(a_{j,n})$, and denote $\lim\limits_{n\to\infty}a_{j,n}$ as $a_j$. According to \cref{eq:sum_on_conv}, we have
		\begin{align}\label{eq:sum_on_conv_seq}
			\lim_{n\to\infty}\sum_{j\in\bbZ\backslash\set{0}}a_{j,n}^k = \sum_{j\in\bbZ\backslash\set{0}}b_{j}^k, \forall k\geq 3.
		\end{align}
		Note that if $k>4$, then the infinite sum $\sum\limits_{j\in\bbZ\backslash\set{0}} \parens*{\dfrac{B}{\abs{j}}}^{\frac{k}{4}}$ converges. This implies that the sum in the \gls{LHS} of \cref{eq:sum_on_conv_seq} converges absolutely, so we can switch the summation with the limit there:
		\begin{align}\label{eq:seq_sum_equal}
			\sum_{j\in\bbZ\backslash\set{0}}a_{j}^k = \sum_{j\in\bbZ\backslash\set{0}}b_{j}^k, \forall k > 4.
		\end{align}
		We are going to prove that $a_{j} = b_j$ for all $j \in\bbZ\backslash\set{0}$ from \cref{eq:seq_sum_equal}. To achieve this, we rearrange the sequences as $(a_{j_l})$ and $(b_{j_l})$ such that $(\abs{a_{j,l}})$ and $(\abs{b_{jl}})$ are non-increasing. Then \cref{eq:seq_sum_equal} can be rewritten as
		\begin{align}\label{eq:seq_rearrange_sum}
			\sum_{l=1}^\infty a_{j_l}^k = \sum_{l=1}^\infty b_{j_l}^k, \forall k > 4.
		\end{align}
		Then it suffices to prove $a_{j_l} = b_{j_l}$. We prove this by induction on $l$.  Suppose we have proved $a_{j_l} = b_{j_l}$ for $l<m$. Then we have 
		\begin{align}\label{eq:abs_sum_induction}
			\sum_{l=m}^\infty \abs{a_{j_l}}^k = \sum_{l=m}^\infty \abs{b_{j_l}}^k,
		\end{align}
		where $k$ is even and $k>4$. We first prove $\abs{a_{j_m}} = \abs{b_{j_m}}$. If $\abs{a_{j_m}} > \abs{b_{j_m}}$, then dividing both sides of \cref{eq:abs_sum_induction} by $\abs{b_{j_m}}$, and let $k\to\infty$ through even numbers, the \gls{LHS} is infinity, and the \gls{RHS} is a finite number, leading to contradiction, hence $\abs{a_{j_m}} \leq \abs{b_{j_m}}$. Similarly it can be shown that $\abs{b_{j_m}} \leq \abs{a_{j_m}}$, thus $\abs{a_{j_m}} = \abs{b_{j_m}}$.
		
		Suppose $b_{j_m}$ appears $p$ times in $(b_j)$ and $q$ times in $(a_j)$; $-b_{j_m}$ appears $p'$ times in $(b_j)$ and $q'$ times in $(a_j)$. Then \cref{eq:seq_sum_equal} can be rewritten as
		\begin{align*}
			(q + (-1)^kq')b_{j_m}^k + \sum_{l>m} a_{j_l}^k = (p + (-1)^kp')b_{j_m}^k + \sum_{l>m} b_{j_l}^k.
		\end{align*}
		Divide both sides by $b_{j_m}^k$ and let $k\to\infty$ through odd numbers we have $q-q'=p-p'$. Similarly by letting $k\to\infty$ through even numbers we have $q+q'=p+p'$. Thus $p=q$ and $p'=q'$, which indicates that $a_{j_m}=b_{j_m}$, which concludes the induction. Therefore, $a_j=b_j$ for all $j \in\bbZ\backslash\set{0}$, which contradicts the assumption that $a_{k_0}\neq b_{k_0}$.
	\end{proof}
	
	\cref{thm:g_freq_conv} implies that the graph frequencies of generalized graphon process asymptotically scale linearly with the square root of number of edges. For a graph sequence converging to a standard graphon, the graph frequencies asymptotically scale linearly with the number of nodes \cite[Lemma 4]{RuiChaRib:J21}. In \cref{sect:exp} we will compare the validity of these hypothesis on a real dataset.
	
	\section{Numerical Experiment}\label{sect:exp}
	
	In this section, we corroborate our results on the ogbn-arxiv dataset\footnote{\url{https://ogb.stanford.edu/docs/nodeprop/}}, where every node represents a paper, and every directed edge represents one paper citing another. We make all edges undirected in this experiment. The entire graph is denoted as $G_{\mathrm{all}} = (V_{\mathrm{all}}, E_{\mathrm{all}})$. We generate a growing graph sequence $(G_n)\subset G$ as follows: 
	\begin{enumerate}
		\item [(1)] we start with an empty graph $\tilde{G}_0$ with no nodes or edges.
		\item [(2)] given $\tilde{G}_n = (\tilde{V}_n,\tilde{E}_n)$, we randomly select $200$ nodes from $V_{\mathrm{all}}\backslash\tilde{V}_n$ without replacement, and add them into $\tilde{V}_n$ to obtain $\tilde{V}_{n+1}$. By letting $\tilde{E}_{n+1} = E_{\mathrm{all}}\bigcap(\tilde{V}_{n+1}\times\tilde{V}_{n+1})$ we obtain $\tilde{G}_{n+1}= (\tilde{V}_{n+1},\tilde{E}_{n+1})$. We iterate this step for $90$ times to get $(\tilde{G}_n)_{n=1}^{90}$.
		\item [(3)] By omitting all isolated vertices in every $\tilde{G}_n$, we obtain the sequence $(G_n)_{n=1}^{90}$.
	\end{enumerate} 
	
	We model the sequence $(G_n)$ as a subsequence of a generalized graphon process $(G_t)$ generated from a graphon, i.e., $(G_n) = (G_{t_n})$ with $t_n\to\infty$. Then according to \cref{thm:g_freq_conv}, $\set{\lambda_j(G_n)\given n=1,2,\dots}$ should have a linear relationship with $\sqrt{\abs{E_n}}$ as $n\to\infty$ for all $j\in\bbZ\backslash\set{0}$. On the other hand, if $(G_n)$ converge to a standard graphon in cut distance, then according to \cite[Lemma 4]{RuiChaRib:J21}, $\set{\lambda_j(G_n)\given n=1,2,\dots}$ should have a linear relationship with $\abs{V_n}$ as $n\to\infty$ for all $j\in\bbZ\backslash\set{0}$. Finally, if $(G_n)$ has bounded degree as assumed by \cite{RodGamBar:J22}, then it can be shown that the set $\set{\lambda_j(G_n)\given n=1,2,\dots,j\in\bbZ\backslash\set{0}}$ is bounded. In order to verify which of these models fits the data best, we fit linear model (with zero interception) for pairs $\set{(\sqrt{\abs{E_n}}, \lambda_j(G_n))}$ and $\set{(\abs{V_n}, \lambda_j(G_n))}$ and test their fitness by \gls{MSE}.
	
	From \cref{tab:MSE} and \cref{fig:fitness_illus} we see that the linear model for $\set{(\sqrt{\abs{E_n}}, \lambda_j(G_n))}$ has better fitness than that for $\set{(\abs{V_n}, \lambda_j(G_n))}$. Due to the sparsity of the graph sequence (see \cref{fig:n_e_ratio}), the standard graphon can be inappropriate as a meaningful limit object. In addition, it appears that the magnitudes of eigenvalues keep increasing instead of clearly bounded by some constant. Therefore, the generalized graphon process model has the best fitness among all models on this dataset.
	
	\begin{table}[!htbp]
		\centering
		\caption{MSE of linear fitness, each data point is obtained by averaging on 20 realizations of the sequence $\set{G_n}$.}
		\label{tab:MSE}
		\begin{tabular}{ |p{4cm}|p{1cm}|p{1cm}|p{1cm}|p{1cm}|p{1cm}|  }
			\hline
			\multicolumn{6}{|c|}{MSE of linear fit} \\
			\hline
			& $\lambda_1$ & $\lambda_2$ & $\lambda_3$ & $\lambda_4$ & $\lambda_5$  \\
			\hline
			Generalized graphon process& 5.98 & 1.74 & 0.91 & 0.54 & 0.43 \\
			Standard graphon &  19.33& 6.83 & 4.49& 3.57& 3.05 \\
			\hline
			& $\lambda_{-1}$ & $\lambda_{-2}$ & $\lambda_{-3}$ & $\lambda_{-4}$ & $\lambda_{-5}$  \\
			\hline
			Generalized graphon process & 6.21 & 1.94 & 1.01 & 0.60& 0.44 \\
			Standard graphon & 18.86 & 6.45 & 4.08& 3.29& 2.68 \\
			\hline
		\end{tabular}
	\end{table}
	
	\begin{figure}[!htbp]
		\centering
		\begin{subfigure}[b]{0.43\columnwidth}
			\centering
			\includegraphics[width=\linewidth, trim=.5cm 0cm 1cm .5cm, clip]{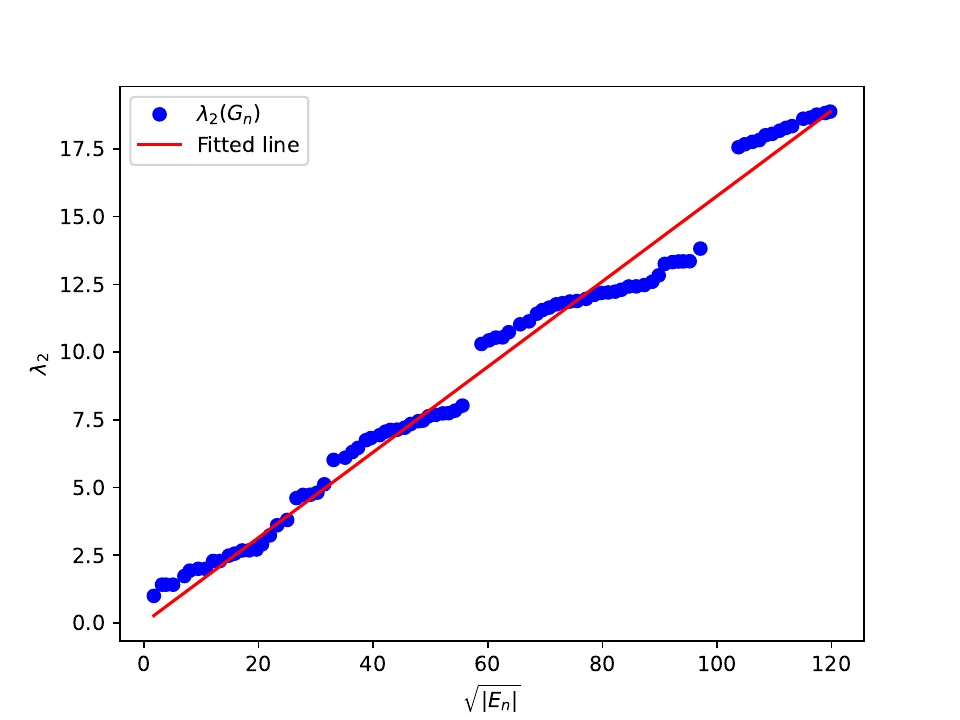}
			\caption{}
			\label{fig:lam_2_edge}
		\end{subfigure}
		\begin{subfigure}[b]{0.43\columnwidth}
			\centering
			\includegraphics[width=\linewidth, trim=.5cm 0cm 1cm .5cm, clip]{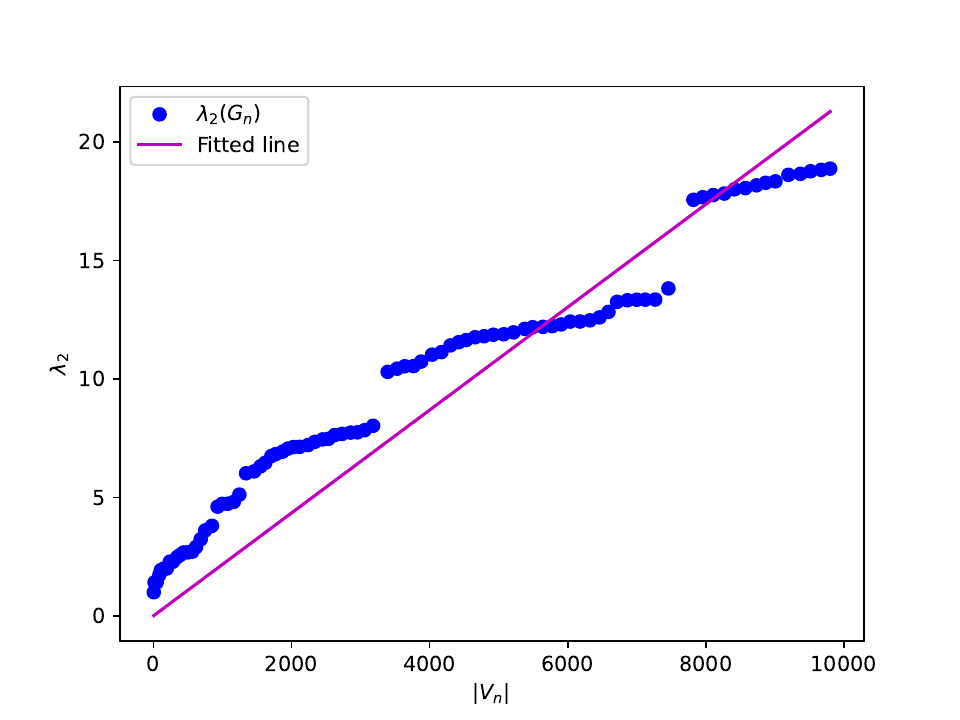}
			\caption{}
			\label{fig:lam_2_node}
		\end{subfigure}
		\begin{subfigure}[b]{0.43\columnwidth}
			\centering
			\includegraphics[width=\linewidth, trim=.5cm 0cm 1cm .5cm, clip]{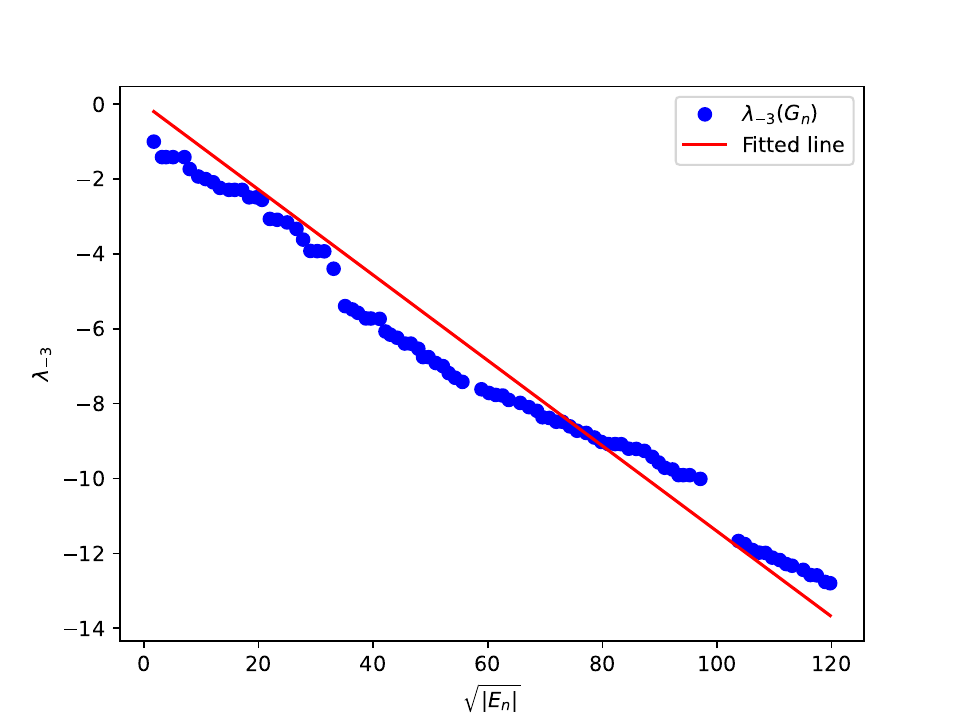}
			\caption{}
			\label{fig:lam-3_node}
		\end{subfigure}
		\begin{subfigure}[b]{0.43\columnwidth}
			\centering
			\includegraphics[width=\linewidth, trim=.5cm 0cm 1cm .5cm, clip]{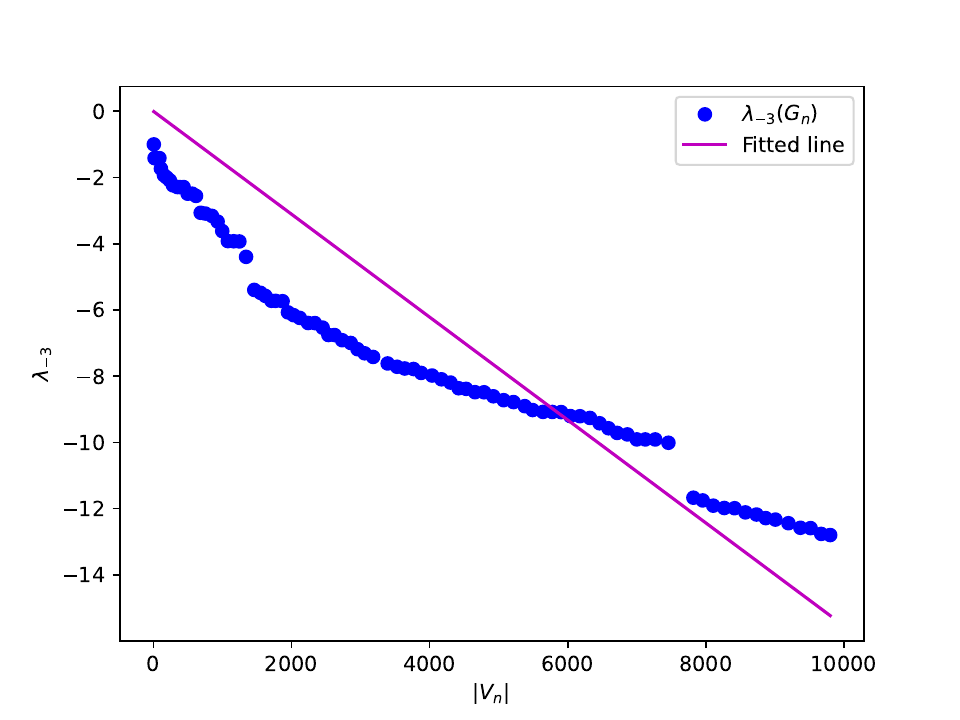}
			\caption{}
			\label{fig:lam-3_edge}
		\end{subfigure}
		\caption{Fitness under generalized graphon process and standard graphon models. }
		\label{fig:fitness_illus}
	\end{figure}
	
	\begin{figure}[!htbp]
		\centering
		\includegraphics[width=0.6\linewidth, trim=.5cm 0cm 1cm .5cm, clip]{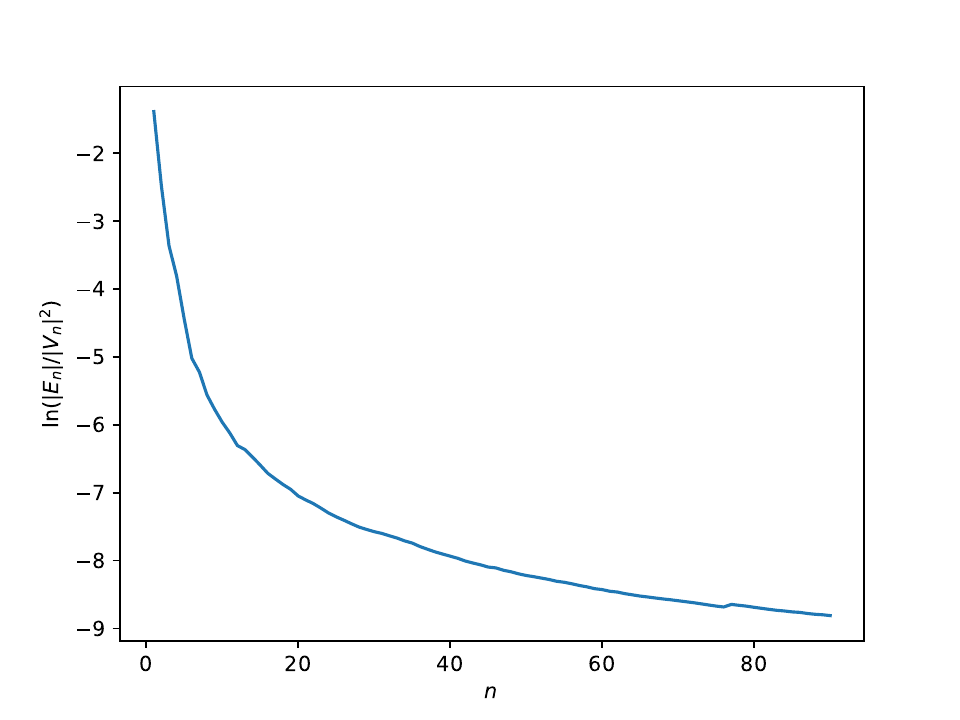}
		\caption{The sparsity of $(G_n)$.}
		\label{fig:n_e_ratio}
	\end{figure}
	
	\section{Conclusion}\label{sect:conc}
	
	In this paper, we have introduced the notions of generalized graphon, stretched cut distance, and generalized graphon process to describe the convergence of sparse graph sequence. To be specific, we studied generalized graphon process that is known to converge to the generalized graphon in stretched cut distance, and proved the convergence of the associated adjacency matrices' eigenvalues, which are known as graph frequencies in \gls{GSP}. This work lays the foundation for transfer learning on sparse graphs. Possible future work includes proving the convergence of \gls{GFT} and filters for generalized graphon process.  
	
	%%% This is normally used for papers instead of bibentry
	
	\bibliographystyle{IEEEtran}
	\bibliography{IEEEabrv,StringDefinitions,refs}

\end{document}